\newif\ifextended\extendedtrue
\newif\ifaec\aecfalse
\newif\ifreview\reviewfalse

\PassOptionsToPackage{svgnames,dvipsnames}{xcolor}
\PassOptionsToPackage{unicode}{hyperref}
\PassOptionsToPackage{naturalnames}{hyperref}

\ifreview
\documentclass[authoryear,acmsmall,screen,anonymous,review]{acmart}
\else
\documentclass[authoryear,acmsmall,screen]{acmart}
\fi

\usepackage{fd}

\title{Understanding Haskell-style Overloading via Open Data and Open Functions}

\author{Andrew Marmaduke}
\orcid{0000-0002-5230-6728}
\affiliation{
  \department{Department of Computer Science}
  \institution{The University of Iowa}
  \streetaddress{14 MacLean Hall}
  \city{Iowa City}
  \state{Iowa}
  \country{USA}}
\email{andrew-marmaduke@uiowa.edu}

\author{Apoorv Ingle}
\orcid{0000-0002-7399-9762}
\affiliation{
  \department{Department of Computer Science}
  \institution{The University of Iowa}
  \streetaddress{14 MacLean Hall}
  \city{Iowa City}
  \state{Iowa}
  \country{USA}}
\email{apoorv-ingle@uiowa.edu}

\author{J. Garrett Morris}
\orcid{0000-0002-3992-1080}
\affiliation{
  \department{Department of Computer Science}
  \institution{The University of Iowa}
  \streetaddress{14 MacLean Hall}
  \city{Iowa City}
  \state{Iowa}
  \country{USA}}
\email{garrett-morris@uiowa.edu}

\begin{document}

\renewcommand\thetheorem{\arabic{theorem}}
\renewcommand\thedefinition{\arabic{definition}}

\begin{abstract}
We present a new, uniform semantics for Haskell-style overloading. We realize our approach in a new core language, \SysFD, whose metatheory we mechanize in the Lean4 interactive theorem prover. \SysFD is distinguished by its open data types and open functions, each given by a collection of instances rather than by a single definition. We show that \SysFD can encode advanced features of Haskell's of type class systems, more expressively than current semantics of these features, and without assuming additional type equality axioms.
\end{abstract}

\maketitle

\section{Introduction}
\InlineOn

One of Haskell's distinguishing features is its support for open, extensible families of terms and types. Type classes \citep{WadlerB89,Kaes88} capture open families of terms, and are used to bring regularity to the typing of ad-hoc polymorphic operations---such as equality, ordering, or conversion to and from text---within a Hindley-Milner type system~\citep{Hindley69,Milner78,DamasMilner82}. For example, the !elem! function---which searches for a probe element in a list---can be given a single, closed definition, generic over the type of list elements:
\begin{codef}
elem :: Eq a => a -> [a] -> Bool
elem x []       = False
elem x (y : ys) = x == y || elem x ys
\end{codef}
To define !elem!, we rely on the equality operator !(==)!.  (The type class predicate !Eq a! assures that !(==)! is defined for !x! and !y!.)  Equality itself cannot have a closed definition: as the programmer adds new data types, they can also extend the equality operator to those types.  This is the sense in which !(==)! is an open family of terms: while its interpretation for a given type in a given program is fixed, its meaning can be extended as the program is extended with new type definitions.

Type families \citep{SchrijversPJCS08} similarly capture open families of types.  Imagine a generalization of the !elem! function to apply to collections other than lists, such as sets, !ByteString!s, or dictionaries.  The type of the probe is different for each of these examples: we would look for an !a! in a !Set a!, a !Char! in a !ByteString!, or a pair !(k, v)! in a !Map k v! dictionary.  We could capture this variation in the type of !elem! by using a type family !Elem!, associating collection types with their element types:
\begin{codef}
elem :: Eq (Elem c) => Elem c -> c -> Bool
\end{codef}
Of course, !elem! itself must now be an open term family.  As the programmer adds new collection types, they can extend both the !Elem! type family and the !elem! term family.

Despite the intuitive parallels between these ideas, type classes and type families are explained by two very different approaches. The semantics of type classes is given by dictionary passing, a type-preserving (nearly) source-to-source translation proposed by \citet{WadlerB89} in their first account of type classes. On the one hand, because dictionary passing is a type-preserving translation, type classes are readily adoptable in other languages (and have been adopted in Agda, Idris, Isabelle, Rust, and more) and experimental type class features cannot compromise the type safety of the core language (the early development of Haskell included a flourishing of type class features, as described by \citet{PeytonJonesJM97}). On the other hand, the expressiveness of type classes is limited by the expressiveness of the target type system. Jones describe functional dependencies~\citep{Jones00} in terms of a fundamental refinement to the interpretation of polymorphic types~\citep{Jones95}. However, functional dependencies have typically been implemented purely in type inference~\citep{DuckPSS04,SulzmannDPJS07,JonesD08}, not as a feature of the underlying type system. As a result, implementations reject programs that are valid in Jones's theory.

Type families are expressed by extension of Haskell's core language, \SysFC: (valid) type families are translated to (trusted) axioms in \SysFC. This means that the expressiveness of type families is unhindered by limitations in the core type system; examples like those that bedeviled functional dependencies are no problem for type families. New features of type families, however, entail new proofs of the type safety of \SysFC, already a non-trivial task.  Moreover, as summarized by \citet{MorrisE17}, many features of type families have surprising  side conditions, such as the reliance on infinitary unification in closed type families \citep{EisenbergVPJW14}, or limitations on the forms of the right-hand sides of injective type family instances \citep{StolarekJE15}.

We propose a new approach to the semantics of type and term families, realized in a new core functional language \SysFD. \SysFD extends System F with three ideas. It includes first-class equality proofs, called \emph{coercions}. (Our treatment of coercions is broadly similar to that of \SysFC.) It includes \emph{open data types}, which can be extended with new cases. Open data types can be used to capture ideas such as instances of type classes and type families or (user-extensible) codes for types. Of course, closed pattern matching is insufficient to express interesting programs over open data types. \SysFD includes \emph{open functions}, which can also be extended with new cases. To capture the interaction of open functions and open data types, \SysFD simulates a form of (guarded) non-determinism by exploring all possibilities. However, just as in Dijkstra's~\citep{Dijkstra75} use of non-determinism in guarded commands, this should be viewed as a formal artifact to capture the generality of patterns that could be expressed with \SysFD open data types and functions, not as a core part of the language.

The goal of \SysFD is to explore new features and new extensions of existing features, not to be a one-for-one replacement for \SysFC as a core language for current Haskell. The latter would require discussing many features essentially orthogonal to \SysFD's novelty---we would have to explain not just type classes and type families, but also GADTs, data kinds, and so forth. Moreover, \SysFC is already custom-fitted to describing this exact collection of features. Instead, we demonstrate a novel account of functional dependencies. Our account realizes functional dependencies as coercions. This means we capture the full expressiveness of Jones's theory, beyond any widely used implementation of functional dependencies. It gives semantic content to the ``consistency'' and ``covering'' checks \citep{Jones00,SulzmannDPJS07,JonesD08} that validate functional dependency declarations. And, it lies outside the image of dictionary-passing translation, even augmented with coercions, demonstrating the additional expressiveness of \SysFD.

We have mechanized the typing and semantics of \SysFD and its core metatheoretic results, including syntactic type safety, in the Lean4 theorem prover. This metatheory is independent of the translation of particular features. That is, in translating features like functional dependencies to \SysFD, we do not extend \SysFD's type system or its axioms, and so cannot compromise type safety. Instead, we rely on open families of coercions. To show that our translation is meaningful, we must argue that these families can be fully and unambiguously populated. This obligation is easily met for our translation, but we would expect it to become more complicated in translating more complicated features.

In mechanizing our metatheory, we have adapted the Autosubst framework \citep{SchaferTS15,StarkSK19} from Rocq to Lean4. Autosubst provides automation for de Bruijn indexing and parallel substitution. Our proofs of progress and preservation require 574 and 639 lines of code, respectively (not counting additional lemmas). Our translations are also implemented in Lean4 programming language fragment, allowing us to validate their results in terms of our core metatheory. In addition to validating the design of \SysFD, we believe this mechanization shows the utility of Lean4 for general programming language metatheory, and will make our adaptation of Autosubst independently available.

In summary, this paper contributes:
\begin{itemize}
\item The design of \SysFD~\cref{sec:system-fd}, a new core functional language that supports first-class type equality proofs, open data types, and open functions;
\item Automated translations of features of Haskell type classes and their uses into \SysFD~\cref{sec:classes}, providing equivalent or greater expressiveness than encodings of these features in other core languages; and,
\item The metatheory of \SysFD~\cref{sec:metatheory}, mechanized in the theorem prover Lean4.
\end{itemize}
We begin by reviewing type classes and type families in Haskell, as well as their standard semantics~\cref{sec:background}, and conclude with discussions of related~\cref{sec:related} and future~\cref{sec:conclusion} work.

\section{Background}
\label{sec:background}

\subsubsection*{Type classes}

Type classes account for ad hoc polymorphism in a parametric polymorphic setting \citep{Kaes88,WadlerB89}. For example, Haskell's !Monad! type class characterizes models of computation, and is defined by:
\begin{code}
class Applicative m => Monad m where
  return :: a -> m a
  (>>=)  :: m a -> (a -> m b) -> m b
\end{code}
Following this definition, a type !M! can be a member of the !Monad! class if (i) it is already a member of the !Applicative! class, and (ii) it has implementations of the !return! and !(>>=)! methods at the claimed types. These definitions must be specific to each type !M!, of course, but other code can then be parametric over any monad:
\begin{code}
sequence         :: Monad m => [m a] -> m [a]
sequence []       = return []
sequence (m : ms) = m >>= \x -> fmap (x :) (sequence ms)
\end{code}
We are able to use !return! and !(>>=)! in this definition because we know that !m! is a member of the !Monad! class, and !fmap! because we know that !m! is a member of the !Applicative! class, and therefore a member of the !Functor! class.

Type classes are given semantics by translation to \emph{dictionary-passing style}. In this translation, class predicates are interpreted as dictionaries---that is, tuples of their implementations---and class methods are interpreted as projections from those tuples. To continue the example, we could interpret the !Monad! class by a type like:
\begin{code}
type Monad m = (Applicative m, forall a. a -> m a, forall a b. m a -> (a -> m b) -> m b)
applicativeM          :: Monad m -> Applicative m
applicativeM (d, _, _) = d
returnM               :: Monad m -> forall a. a -> m a
returnM (_, ret, _)    = ret
bindM                 :: Monad m -> forall a b. m a -> (a -> m b) -> m b
bindM (_, _, bind)     = bind
\end{code}
An instance of the !Monad! class would then be interpreted as a suitable tuple, and code that used !Monad!s would be translated to take these tuples as arguments:
\begin{code}
sequence :: Monad m -> [m a] -> m [a]
sequence d [] = returnM d []
sequence d (m : ms) =
  bindM d m (\x -> fmap (functorA (applicativeM d))) (x :) (sequence d ms)
\end{code}

Many examples of type classes are more easily thought of as relations among types rather than properties of a single type. For example, to describe models of stateful computation, we need to know not just the type of the model but the type of the state as well. These cases were originally captured with multi-parameter type classes:
\begin{code}
class Monad m => MonadState s m where
  get :: m s
  put :: s -> m ()
\end{code}
The translation of multi-parameter type classes is identical to that of single-parameter type classes.

\subsubsection*{Type families}

Unfortunately, the !MonadState! example leads to counterintuitive types in practice. To see why, consider two (relatively uninteresting) stateful operations:
\begin{code}
getPut   = get   >>= put
putGet x = put x >>= \ () -> get
\end{code}
The most general type of !getPut! is what we might expect:
\begin{code}
getPut :: MonadState s m => m ()
\end{code}
The most general type of !putGet!, however, is less intuitive:
\begin{code}
putGet :: (MonadState s m, MonadState t m) => s -> m t
\end{code}
In fact, both of these types are problematic, but for different reasons. The type of !getPut! is problematic because the type parameter !s! appears in the constraints---and so the choice of !s! can influence the semantics of !getPut!---but does not appear in the type itself. Such a type parameter is called \emph{ambiguous}, and is considered a type error. The !putGet! function is not ill-typed; however, it is still likely counterintuitive that the type of state installed by the !put! operation be completely unrelated to the type of state retrieved by the !get! operation.

The missing component of this definition is that in asserting !MonadState s m!, we likely intended to say that !s! is the \emph{unique} type of state modeled by !m!, not just that !s! is \emph{one} type of state modeled by !m!. Associated types \citep{ChakravartyKPJM05} allow us to tighten the connection:
\begin{code}
class Monad m => MonadState m where
  type State m
  get :: m (State m)
  put :: State m -> m ()
\end{code}
Rather than the state being an arbitrary type variable, it can now be seen as a function of the type !m!. With this definition, !getPut! and !putGet! are both well-typed:
\begin{code}
getPut :: MonadState m => m ()
putGet :: MonadState m => State m -> m (State m)
\end{code}

While type families, like !State!, were initially proposed as part of the class system, it was soon realized that they were an independent feature \citep{SchrijversPJCS08}. While this makes little difference to the example above---the !State! family is useful only in the context of a !MonadState! instance, and vice versa---separating type families from type classes both clarified their semantics (they are not explained by dictionary passing) and allowed them to evolve independently of the class system. Among these extensions are: closed type families \citep{EisenbergVPJW14}, which allow overlapping type family instances, and can be used to express conditionals; injective type families \citep{StolarekJE15}, which make it possible to capture properties like duality as type families; and kind families \citep{WeirichHE13}, which enable richer type-level programming.

Type families are given semantics in \SysFC, a core functional language with first-class type equality proofs called \emph{coercions}. Instances of type families are interpreted as new coercions axioms in \SysFC. For example, an instance of the !MonadState! class above might begin
\begin{code}
instance MonadState (StateM s) where
  type State = s
  ...
\end{code}
where a !StateM s a! wraps a function !s -> (a, s)!. This instance would correspond to a new coercion axiom !State s ~ s! in \SysFC. Extension of type families have corresponding extensions to axioms in \SysFC: closed type families give rise to overlapping chains of axioms, injective type families give rise to injectivity axioms. As these axioms are trusted in \SysFC, the type safety of the resulting language depends not just on the core of \SysFC, but also on the axioms allowed. Thus, more complicated features of type families lead to correspondingly complex type safety proofs, even as the core of \SysFC remains (relatively) simple.

\subsubsection*{Functional dependencies}

Type families were the second mechanism in Haskell to address the difficulties with multi-parameter type classes, preceded by functional dependencies \citep{Jones00}. Functional dependencies have a similar intuition: they annotate classes themselves with intended functional relationships among their parameters. The example above would be annotated:
\begin{code}
class Monad m => MonadState s m | m -> s where ...
\end{code}
The annotation means that, if the constraints !MonadState s m! and !MonadState t m! are both satisfiable, then !s! must be the same type as !t!. With functional dependencies, !putGet! and !getPut! are again well-typed:
\begin{code}
getPut :: MonadState s m => m ()
putGet :: MonadState s m => s -> m s
\end{code}
The type of !getPut! is no longer ambiguous, as !s! is determined by !m! and !m! is unambiguous. As the !put! and !get! operations in !putGet! share the same monad !m!, we can also be sure that they share the same state type !s!.

Unlike type families, functional dependencies have generally been treated as a feature of type inference \citep{DuckPSS04,SulzmannDPJS07,JonesD08}, not as a feature of Haskell's core language or type system. Concretely, the type of !putGet! inferred with functional dependencies is an instantiation of the type inferred without them; the role of the functional dependencies is to guarantee that this apparently more specific type is not actually any less general. \citet{Morris14} proposes a typing rule that uses functional dependencies to justify generalizing types, but his rule is justified only by a denotational approach to the semantics of Haskell.

\section{\SysFDbf}
\label{sec:system-fd}

This section introduces the syntax and typing of \SysFD; the following sections show that \SysFD naturally and expressively captures many features of the Haskell type and class systems~\cref{sec:classes} and present its formal semantics and metatheory, mechanized in the Lean4 theorem prover~\cref{sec:metatheory}. \SysFD differs from other core functional languages in three essential ways:
\begin{itemize}
\item \SysFD includes first-class type equality proofs, or coercions, inspired by those of \SysFC.
\item \SysFD includes open data types, which may have new constructors introduced throughout a program.
\item \SysFD includes open functions, given by a set of guarded definitions throughout a program.
\end{itemize}
Open functions and open data types are closely linked. Values of an open datatype could not be (safely) eliminated using traditional branching constructs, as the open datatype may have more constructors when the eliminator is used than it did when the eliminator was defined. Instead, open data types are eliminated using open functions. In turn, open functions are (usefully) given by collections of definitions, each specific to one constructor of an open data type.

\InlineOff
\begin{figure}
\begin{smalle}
\[
\begin{array}{l@{\hspace{5px}}l@{\qquad}l@{\hspace{5px}}l@{\qquad}l@{\hspace{5px}}l@{\qquad}l@{\hspace{5px}}l}
  \text{Term variables} & x, y & \text{Type variables} & t, u & \text{Term constants} & K & \text{Type constants} & T, C
\end{array}
\]
\begin{doublesyntax}
  \mcl{\text{Kinds}} & \kappa & ::= & \Type \mid \kappa \to \kappa \\
  \text{Types} & \mcr{\sigma, \tau, \upsilon} & ::= & t \mid T \mid \tau \Ap \upsilon \mid \Eq \tau \upsilon \kappa \mid \ForallS t \kappa \sigma \\
  \text{Terms} & \mcr{L, M, N, \eta} & ::= & x \mid K \mid \LamS x \sigma M \mid M \Ap N \mid \TLam t \kappa M \mid M \TAp \sigma \mid \Cast M \eta \\
  & & & \mid & \Case L P M N \mid \Guard M P N \mid \Zero \mid \Choice M N\\
  \mcl{\text{(Coercions)}} & & \mid & \Refl \tau \mid \Sym \eta \mid \Trans \eta \eta \mid \CAp \eta {\eta'} \mid \Fst \eta \mid \Snd \eta \mid \Univ t \kappa \eta \mid \CInst \eta \tau \mid \SimC \eta {\eta'} \\
  \mcl{\text{Patterns}} & P & ::= & K \mid P \TAp \sigma \\
  \mcl{\text{Declarations}} & D, E & ::= & \DataDecl T \kappa \mid \CtorDecl K \sigma \mid \OpenDecl C \kappa \mid \OpenCtorDecl K \sigma \\
  & & & \mid & \MethodDecl x \tau \mid \InstanceDecl x M \mid \LetDecl x \sigma M \\
  \mcl{\text{Environments}} & \Gamma,\Delta & ::= & \Empty \mid \Gamma, t : \kappa \mid \Gamma, x : \sigma \mid \Gamma, \DataSig T \kappa \mid \Gamma , \OpenSig C \kappa \mid \Gamma , K : \sigma \mid \Gamma, \OpenSig x \sigma  \\
  \mcl{\text{(Runtime)}} & & \mid & \Gamma, \InstanceDef x M \mid \Gamma, \LetDef x M
\end{doublesyntax}
\end{smalle}
\caption{Syntax}
\label{fig:syntax}
\end{figure}

\subsubsection*{Syntax.}

\cref{fig:syntax} gives the syntax of \SysFD. $\Eq \tau \upsilon \kappa$ is the type of coercions from $\tau$ to $\upsilon$, and are annotated with the kind $\kappa$ of $\tau$ and $\upsilon$. We will assume that we always have access to a type constructor $(\to) : \star \to \star \to \star$, which we will typically write infix. $\Cast M \eta$ casts $M$ by the coercion $\eta$: if $M$ has type $\sigma$, and $\eta$ has type $\Eq \sigma \tau \kappa$, then $\Cast M \eta$ has type $\tau$. Our grammar of coercions is similar to that in \SysFC. The values of coercion type are trees of choices with leaves of $\Refl \tau$ of type $\Eq \tau \tau \kappa$; in particular, unlike in semantics of type families in \SysFC, we will not introduce new coercion axioms. Pattern matching over closed data types is expressed with $\mathsf{if}$: if the scrutinee $L$ is built with the constructor identified by $P$, evaluation continues with the consequent $M$; otherwise, it continues with $N$. Pattern matching over open data types is expressed with $\mathsf{guard}$. A $\mathsf{guard}$ has no $\mathsf{else}$ branch; instead, other cases must be handled by separate (guarded) definitions. Finally, programs consist of lists of declarations. Conceptually, the declaration of a closed datatype ($\DataDecl T \kappa$) and the definition of its constructors ($\CtorDecl K \sigma$) are linked; we separate them only to simplify the presentation. In contrast, to be of any interest, the declaration of an open datatype ($\OpenDecl C \kappa$) or open function ($\MethodDecl x \sigma$) must be separate from the definitions of its constructors ($\OpenCtorDecl K \sigma$) or instances ($\InstanceDecl x M$). Environments record not just the types of open functions and \lstinline!let!-bound variables, but also their definitions. The latter are needed to define the reduction relation, but not to define typing.

\begin{figure}
\begin{smalle}
\[
\begin{gathered}[t]
\fbox{$\EnvJ \Gamma$}\\
\ib{\irule{ };
          {\EnvJ\Empty}}
\isp
\ib{\irule{\EnvJ \Gamma}
          {\KindJ \Gamma \sigma \Type};
          {\EnvJ {\Gamma, x : \sigma}}}
\isp
\ib{\irule{\EnvJ \Gamma}
          {\KindJ \Gamma \sigma \Type};
          {\EnvJ {\Gamma, K : \sigma}}}
\isp
\ib{\irule{\EnvJ \Gamma};
          {\EnvJ {\Gamma, t : \kappa}}}
\\
\ib{\irule{\EnvJ \Gamma};
          {\EnvJ {\Gamma, \DataSig T \kappa}}}
\isp
\ib{\irule{\EnvJ \Gamma};
          {\EnvJ {\Gamma, \OpenSig C \kappa}}}
\isp
\ib{\irule{\KindJ \Gamma \sigma \Type};
          {\EnvJ {\Gamma, \OpenSig x \sigma}}}
\\
\ib{\irule{\OpenSig x \sigma \in \Gamma}
          {\TypeJ \Gamma M \sigma};
          {\EnvJ {\Gamma, \InstanceDef x M}}}
\isp
\ib{\irule{x : \sigma \in \Gamma}
          {\TypeJ \Gamma M \sigma};
          {\EnvJ {\Gamma, \LetDef x M}}}
\end{gathered}
\qquad
\begin{gathered}[t]
\fbox{$\KindJ \Gamma \sigma \kappa$}
\\
\ib{\irule{\EnvJ \Gamma}
          {t : \kappa \in \Gamma};
          {\KindJ \Gamma t \kappa}}
\isp
\ib{\irule{\KindJ \Gamma \tau {\kappa'\to\kappa}}
          {\KindJ \Gamma \upsilon {\kappa'}};
          {\KindJ \Gamma {\tau\Ap\upsilon} \kappa}}
\\
\ib{\irule{\KindJ {\Gamma, t : \kappa} \tau \Type};
          {\KindJ \Gamma {\ForallS t \kappa \sigma} \Type}}
\isp
\ib{\irule{\KindJ \Gamma \tau \kappa}
          {\KindJ \Gamma \upsilon \kappa};
          {\KindJ \Gamma {\Eq \tau \upsilon \kappa} \Type}}
\\
\ib{\irule{\EnvJ \Gamma}
          {\DataSig T \kappa \in \Gamma \lor \OpenSig T \kappa \in \Gamma};
          {\TypeJ \Gamma T \kappa}}
\end{gathered}
\]
\end{smalle}
\caption{Environment formation and kinding} %
\label{fig:env-kind-typequ}
\end{figure}

\subsubsection*{Types and kinds.}

\cref{fig:env-kind-typequ} gives the rules for environment formation and kinding. We do not distinguish between closed and open types at the kind level, avoiding the need to discuss subkinding or kind polymorphism. Coercions are allowed between types of arbitrary kind, not just of base kind. We assume that quantified types are identified up to $\alpha$-equivalence, and this is the only non-trivial case of type equality. (Our mechanized formalization uses de Bruijn representations of binders, so this point becomes irrelevant.)

\begin{figure}
\begin{smalle}
\begin{gather*}
\fbox{$\TypeJ \Gamma M \sigma$}
\\
\ib{\irule{\EnvJ \Gamma}
          {x : \sigma \in \Gamma};
          {\TypeJ \Gamma x \sigma}}
\rsp
\ib{\irule{\EnvJ \Gamma}
          {K : \sigma \in \Gamma};
          {\TypeJ \Gamma K \sigma}}
\rsp
\ib{\irule{\EnvJ \Gamma}
          {\OpenSig x \sigma \in \Gamma};
          {\TypeJ \Gamma x \sigma}}
\rsp
\ib{\irule{\TypeJ {\Gamma, x : \sigma} M \tau};
          {\TypeJ \Gamma {\LamS x \sigma M} {\sigma\to\tau}}}
\\
\ib{\irule{\TypeJ \Gamma M {\tau\to\upsilon}}
          {\TypeJ \Gamma N \tau};
          {\TypeJ \Gamma {M \Ap N} \upsilon}}
\rsp
\ib{\irule{\TypeJ {\Gamma, t : \kappa} M \sigma};
          {\TypeJ \Gamma {\TLam t \kappa M} {\ForallS t \kappa \sigma}}}
\rsp
\ib{\irule{\TypeJ \Gamma M {\ForallS t \kappa \sigma}}
          {\KindJ \Gamma \tau \kappa};
          {\TypeJ \Gamma {M \TAp \tau} {\Subst t \tau \sigma}}}
\rsp
\ib{\irule{\TypeJ \Gamma M \tau}
          {\TypeJ \Gamma \eta {\Eq \tau \upsilon \Type}};
          {\TypeJ \Gamma {\Cast M \eta} \upsilon}}
\\
\ib{\irule{\TypeJ \Gamma L \sigma}
          {\DataJ \Gamma \sigma}
          {\TypeJ \Gamma P {\forall \overline{t:\kappa}. \overline\tau \to \sigma}}
          {\TypeJ \Gamma M {\forall \overline{t:\kappa}. \overline\tau \to \upsilon}}
          {\TypeJ \Gamma N \upsilon};
          {\TypeJ \Gamma {\Case L P M N} \upsilon}}
\rsp
\ib{\irule{};
          {\TypeJ \Gamma \Zero \sigma}}
\\
\ib{\irule{\TypeJ \Gamma L \sigma}
          {\OpenJ \Gamma \sigma}
          {\TypeJ \Gamma P {\forall \overline{t:\kappa}. \overline\tau \to \sigma}}
          {\TypeJ \Gamma M {\forall \overline{t:\kappa}. \overline\tau \to \upsilon}};
          {\TypeJ \Gamma {\Guard L P M} \upsilon}}
\rsp
\ib{\irule{\TypeJ \Gamma M \sigma}
          {\TypeJ \Gamma N \sigma};
          {\TypeJ \Gamma {\Choice M N} \sigma}}
\end{gather*}
\[
\begin{gathered}
\fbox{$\DataJ \Gamma \tau$}\\
\ib{\irule{\DataSig T \kappa \in \Gamma};
          {\DataJ \Gamma T}}
\rsp
\ib{\irule{\DataJ \Gamma \tau};
          {\DataJ \Gamma {\tau\Ap\upsilon}}}
\end{gathered}
\qquad\qquad
\begin{gathered}
\fbox{$\OpenJ \Gamma \tau$}\\
\ib{\irule{\OpenSig C \kappa \in \Gamma};
          {\OpenJ \Gamma C}}
\rsp
\ib{\irule{\OpenJ \Gamma \tau};
          {\OpenJ \Gamma {\tau\Ap\upsilon}}}
\end{gathered}
\]
\end{smalle}
\caption{Typing: non-coercion terms.}
\label{fig:typing}
\end{figure}
\InlineOn
\subsubsection*{Terms.}

Typing of non-coercion terms is given in \cref{fig:typing}. The typing rules for term and type abstractions are standard. Conditionals and guards are typed similarly; we will focus on the conditional case. For an example, consider the Haskell definition:
\begin{code}
mapMaybe :: (a -> b) -> Maybe a -> Maybe b
mapMaybe f (Just x) = Just (f x)
mapMaybe f Nothing  = Nothing
\end{code}
This function would be expressed in \SysFD by:

\begin{code}
/\ a b. \ f m. if m is Just [a] then \ x. Just [a] (f x) else Nothing [a]
\end{code}

\noindent
(We omit type annotations on type and term arguments, as they can be reconstructed from context.) To express the pattern match, we begin by testing whether $m$ was constructed by !Just [a]!. Note that patterns are not constructors alone, but constructors applied to type arguments. We then require the consequent to have the type !a -> Maybe b!, that is, we align the type of the consequent with the type of the pattern (!a -> Maybe a!) but replacing the codomain with the result type of the $\mathsf{if}$ expression. The alternative is required to have type !Maybe b!. Guards are typed identically, but without an alternative branch. We will see many examples of guards in the following section. Our typing rules for the consequent branches of guards and conditions limits the need to account for binders to type and term abstractions, simplifying the presentation overall. Conditionals are limited to apply to closed data types (via the judgment $\DataJ \Gamma \sigma$) and guards are limited to apply to open data types (via the judgment $\OpenJ \Gamma \sigma$); we do not insist that pattern matches over closed data types be exhaustive, but doing so would be an entirely routine extension of our rules.
\InlineOff
\subsubsection*{Guards and non-determinism.}

Pattern matching over open data types uses open functions. Intuitively, an open function is given by a collection of definitions, each of which is guarded on one or more values of an open data type. We formally capture this intuition using non-determinism. Our design parallels Dijkstra's \citep{Dijkstra75} guarded commands, which he interpreted via a non-deterministic conditional expression $(\mathtt{if} \; B_1 \to M_1 \talloblong \cdots \talloblong B_n \to M_n \; \mathtt{fi})$. The term $\Choice M N$ denotes the non-deterministic evaluation of $M$ and $N$, choosing whichever one does not fail. A guard that does not match reduces to failure, represented by $\Zero$. An open function evaluates to a choice between its definitions, all of which except one should promptly reduce to $\Zero$. Of course, we could have imposed this structure on open functions more explicitly, and in doing so removed the need for the $\Choice M N$ and $\Zero$ term forms. However, doing so would have complicated our mechanization, and might restrict future uses of \SysFD. Moreover, it would supply no actual metatheoretic guarantees: Even with such restrictions, translations would still have to guarantee that they produced open functions without too many or too few cases.

\subsubsection*{Coercions.}

\begin{figure}
\begin{smalle}
\begin{gather*}
\fbox{$\CoerceJ \Gamma \eta \tau \upsilon \kappa$}
\\
\ib{\irule{\KindJ \Gamma \tau \kappa};
          {\CoerceJ \Gamma {\Refl \tau} \tau \tau \kappa}}
\rsp
\ib{\irule{\CoerceJ \Gamma \eta \upsilon \tau \kappa};
          {\CoerceJ \Gamma {\Sym \eta} \tau \upsilon \kappa}}
\rsp
\ib{\irule{\CoerceJ \Gamma \eta \sigma \tau \kappa}
          {\CoerceJ \Gamma {\eta'} \tau \upsilon \kappa};
          {\CoerceJ \Gamma {\Trans \eta {\eta'}} \sigma \upsilon \kappa}}
\\
\ib{\irule{\CoerceJ \Gamma \eta \tau {\tau'} {\kappa' \to \kappa}}
          {\CoerceJ \Gamma {\eta'} \upsilon {\upsilon'} {\kappa'}};
          {\CoerceJ \Gamma {\CAp \eta {\eta'}} {\tau \Ap \upsilon} {\tau' \Ap \upsilon'} \kappa}}
\rsp
\ib{\irule{\KindJ \Gamma \tau \kappa}
          {\CoerceJ \Gamma \eta {\tau\Ap\upsilon} {\tau'\Ap\upsilon'} {\kappa'}};
          {\CoerceJ \Gamma {\Fst \eta} \tau {\tau'} \kappa}}
\rsp
\ib{\irule{\KindJ \Gamma \upsilon \kappa}
          {\CoerceJ \Gamma \eta {\tau\Ap\upsilon} {\tau'\Ap\upsilon'} {\kappa'}};
          {\CoerceJ \Gamma {\Snd \eta} \upsilon {\upsilon'} \kappa}}
\\
\ib{\irule{\CoerceJ {\Gamma, t : \kappa} \eta \tau \upsilon {\kappa'}};
          {\CoerceJ \Gamma {\Univ t \kappa \eta} {\ForallS t \kappa \tau} {\ForallS t \kappa \upsilon} {\kappa'}}}
\rsp
\ib{\irule{\CoerceJ \Gamma \eta {\ForallS t \kappa \tau} {\ForallS t \kappa \upsilon} {\kappa'}}
          {\KindJ \Gamma \sigma \kappa};
          {\CoerceJ \Gamma {\CInst \eta \sigma} {\Subst t \sigma \tau} {\Subst t \sigma \upsilon} {\kappa'}}}
\rsp
\ib{\irule{\CoerceJ \Gamma \eta \tau {\tau'} \kappa}
          {\CoerceJ \Gamma {\eta'} \upsilon {\upsilon'} \kappa};
          {\CoerceJ \Gamma {\SimC \eta {\eta'}} {(\Eq \tau \upsilon \kappa)} {(\Eq {\tau'} {\upsilon'} \kappa)} \Type}}
\end{gather*}
\end{smalle}
\caption{Typing: coercions}
\label{fig:typing-coercions}
\end{figure}

\cref{fig:typing-coercions} gives the typing rules for coercions. (We separate coercions from the other terms solely for presentation.) The only coercion values are reflexivity $\Refl\tau$ for all well-kinded types. The rules for symmetry, transitivity, and type applications are all identical to those in \SysFC. We have a less elaborate treatment of coercions for quantified types, as we choose to build $\alpha$-equivalence into type equality rather than capturing it with coercions. Here we have two motivations. First, were $\alpha$-equivalence not built into type equality, then a coercion between two $\alpha$-equivalent quantified types could not reduce to reflexivity. Second, our mechanization uses a de Bruijn representation of quantifiers, where $\alpha$-equivalent quantified types are syntactically identical. Finally, we include coercions for the primitive function and coercion types.

\begin{figure}
\begin{smalle}
\begin{gather*}
\fbox{$\DeclJ \Gamma D \Delta$}
\\
\ib{\irule{\vdash \Gamma};
          {\DeclJ \Gamma {\DataDecl T \kappa} {\Gamma,\DataSig T \kappa}}}
\rsp
\ib{\irule{\vdash \Gamma};
          {\DeclJ \Gamma {\OpenDecl C \kappa} {\Gamma,\OpenSig C \kappa}}}
\rsp
\ib{\irule{\KindJ \Gamma \sigma \Type};
          {\DeclJ \Gamma {\MethodDecl x \sigma} {\Gamma, \OpenSig x \sigma}}}
\\
\ib{\irule{\KindJ \Gamma \sigma \Type}
          {\DataJ \Gamma \upsilon};
          {\DeclJ \Gamma {\CtorDecl K \sigma} {\Gamma, K : \sigma}}!
          {\sigma = {\forall {\overline{t : \kappa}}. \overline\tau \to \upsilon}}}
\rsp
\ib{\irule{\KindJ \Gamma \sigma \Type}
          {\OpenJ \Gamma \upsilon};
          {\DeclJ \Gamma {\OpenCtorDecl K \sigma} {\Gamma, K : \sigma}}!
          {\sigma = {\forall {\overline{t : \kappa}}. \overline\tau \to \upsilon}}}
\\
\ib{\irule{\OpenSig x \sigma \in \Gamma}
          {\TypeJ \Gamma M \sigma};
          {\DeclJ \Gamma {\InstanceDecl x M} {\Gamma, \InstanceDef x M}}}
\rsp
\ib{\irule{\TypeJ \Gamma M \sigma};
          {\DeclJ \Gamma {\LetDecl x \sigma M} {\Gamma, x : \sigma, \LetDef x M}}}
\end{gather*}
\end{smalle}
\caption{Typing: declarations}
\label{fig:typing-declarations}
\end{figure}

\subsubsection*{Declarations}

\cref{fig:typing-declarations} gives the typing rules for declarations; the judgment $\DeclJ \Gamma D \Delta$ denotes that $D$ is well-typed in environment $\Gamma$, generating environment $\Delta$. We insist that instances of open functions all have identical types---that is, that the instances are all as general as the initial declaration. Type-specific behavior in open functions is enabled by guards within individual definitions.

\section{Encoding Type Classes and Functional Dependencies in \SysFDbf}
\label{sec:classes}
\InlineOn

This section demonstrates the encoding of type classes and functional dependencies in \SysFD.
We describe a simplified surface language, sufficient to demonstrate the definition and use of type classes and functional dependencies.
We then illustrate the translation of our surface language to \SysFD.
Finally, we describe our implementation of the translation, and identify important properties of the translation.

\subsection{Surface Language}

\begin{figure}
\begin{smalle}
\[
\begin{array}{l@{\hspace{5px}}l@{\qquad}l@{\hspace{5px}}l@{\qquad}l@{\hspace{5px}}l}
  \text{Type variables} & t, u & \text{Type constants} & T, C & \text{Indices} & m, n \\
  \text{Term variables} & x, y & \text{Term constants} & K
\end{array}
\]
\begin{doublesyntax}
  \mcl{\text{Kinds}} & \kappa & ::= & \Type \mid \kappa \to \kappa \\
  \text{Types} & \mcr{\sigma, \tau, \upsilon} & ::= & t \mid T \mid \tau \Ap \upsilon \mid \ForallS t \kappa \sigma \\
  \text{Terms} & \mcr{L, M, N} & ::= & x \mid K \mid \LamS x \sigma M \mid M \Ap N \mid \TLam t \kappa M \mid M \TAp \sigma \mid \Case L P M N\\
  & & & \mid & \Hole \sigma \mid \Annotate t \sigma \\
  \mcl{\text{Patterns}} & P & ::= & K \mid P \TAp \sigma \\
  \mcl{\text{Declarations}} & D, E & ::= & \DataDecl T \kappa\ (\CtorDecl K \sigma; \ldots) \\
  & & & \mid & \ClassDecl C \kappa\ (\mathsf{super} \, \sigma; \ldots)\ (\FunDep {m, \dots} n; \ldots)\ (\ClassMethodDecl k \tau; \ldots)\\
  & & & \mid & \ClassInstanceDecl C\ (\ClassMethodInstanceDecl M; \ldots) \mid \LetDecl x \sigma M
\end{doublesyntax}
\end{smalle}
\caption{Surface language syntax}
\label{fig:surface-syntax}
\end{figure}

One might expect our surface language to be just Haskell 98 with type class extensions.
However, describing the translation of Haskell 98 would involve many orthogonal issues, such as Hindley-Milner style type inference and generalization.
To focus on the novelty of \SysFD, we choose a surface language with explicit polymorphism, but in which instance arguments and functional dependency coercions are still implicit.

Our surface language is shown in \cref{fig:surface-syntax}. The term language is a subset of \F, with two additional features. Annotations ($\Annotate t \sigma$) ascribe types to terms. While they can appear anywhere, our translation requires annotations in two specific cases (1) on a non-variable head of application spines, and (2) on the pattern, scrutinee, and consequent of conditional expressions.
This is simply to reduce the complexity of the translation, and does not represent any fundamental
limitation. Holes ($\Hole \sigma$) appear where instance arguments need to be computed. Again to simplify the translation, we require that function terms are fully $\eta$-expanded. For example, \cref{fig:h98-surface} shows a Haskell-style definition on the top and its representation in our surface language at the bottom.

\begin{figure}[ht]
\centering
\begin{code}
example :: Eq a => (Bool -> a) -> a -> Bool
example x y = eq (x ((\ z. z) True)) y
\end{code}
\begin{code}
let example :: forall a. Eq a => (Bool -> a) -> a -> Bool
  = /\ a :: *. \ e :: Eq a. \ x :: Bool -> a. \ y :: a.
     eq [a] (_ :: Eq a) (x (((\ z :: Bool. z) :: Bool -> Bool)) True) y
\end{code}
  \caption{Representing Haskell 98 in a Surface Language}
  \label{fig:h98-surface}
\end{figure}

\noindent
We have fully $\eta$-expanded the term, making type and instance arguments explicit.
We have identified where an instance argument is required, by the !eq! function, but we do not have to provide that argument.
Furthermore, we do not have to specify the locations of coercions; their use will be inferred automatically in translation when types fail to align.

Declarations include data types, classes, class instances, and normal term definitions.
Data types are equipped with a list of their constructors.
Classes are equipped with their kind and lists of superclasses, functional dependencies, and class methods.
A functional dependency includes the set of determining parameters and the determined parameter, represented as indices.
For example, the !MonadReader! class would have kind !* -> (* -> *) -> *!, and the functional dependencies !1 -> 0!.
Class instances include their set of method definitions.
Note that class instances do not have to say anything about their functional dependencies.

\subsection{Type Classes}

We start with an example simplified from the \lstinline!Eq! and \lstinline!Ord! classes in Haskell's prelude. This example demonstrates the parallels between interpreting type classes via dictionaries and via open data types. The left-hand column of \cref{fig:eq-ord} defines each class, along with one method for each class. We also declare a superclass relationship between the classes: every instance of \lstinline!Ord! must also be an instance of \lstinline!Eq!. We use a Haskell style syntax of the corresponding surface syntax\footnote{\protect\url{Hs/Examples/SuperClasses.lean}} to keep the presentation lucid.

\begin{figure}
\begin{minipage}[t]{0.38\textwidth}
\begin{codef}
class Eq a where
  (==) :: a -> a -> Bool
instance Eq Bool where
  (==) = (\ b c -> not (b `xor` c))
          :: Bool -> Bool -> Bool
\end{codef}
\end{minipage}
\begin{minipage}[t]{0.60\textwidth}
\begin{codef}
open Eq :: * -> *
open eq :: forall a. Eq a -> a -> a -> Bool
instance EqBool :: forall t. Bool ~ t -> Eq t
instance eq = /\ a. \ d. guard d is EqBool [a] then
  \ h. (\ b c. not (b `xor` c)) |>
    $\Refl{\to}$ @ h @ ($\Refl\to$ @ h @ $\Refl{\mathsf{Bool}}$)
\end{codef}
\end{minipage}
\\
\begin{minipage}[t]{0.38\textwidth}
\begin{codef}
class Eq a => Ord a where
  (<) :: a -> a -> Bool

instance Ord Bool where
  (<) = (\ b c -> not b || c)
         :: Bool -> Bool -> Bool
\end{codef}
\end{minipage}
\begin{minipage}[t]{0.60\textwidth}
\begin{codef}
open Ord :: * -> *
open ordEq :: forall a. Ord a -> Eq a
open lt :: forall a. Ord a -> a -> a -> Bool
instance OrdBool :: forall t. Bool ~ t -> Ord t
instance ordEq = /\ a. \ d.
  guard d is OrdBool [a] then EqBool [a]
instance lt = /\ a. \ d. guard d is OrdBool [a] then
  \ h. (\ b c. not b || c) |>
    $\Refl{\to}$ @ h @ ($\Refl\to$ @ h @ $\Refl{\mathsf{Bool}}$)
\end{codef}
\end{minipage}
\caption{\SysFD translation of simple type classes}
\label{fig:eq-ord}
\end{figure}

\subsubsection*{Classes and methods}

The right-hand column contains the translations into \SysFD\footnote{\protect\url{SystemFD/Examples/SuperClasses.lean}}. We omit annotations on type and term abstractions, as they can be inferred from context. The translation of the !Eq! class itself is straightforward: we interpret the class itself as an
open type definition, and its method as an open function definition. Note that the kind of the !Eq! type is a direct transliteration of the surface language type while the type of the !(==)! method is the expected translated type where the implicit !Eq a! type in the surface level is made explicit in the \SysFD type.

\subsubsection*{Instances}

As an intuition for the translation of instances, imagine that we were using closed data types to represent a fixed set of instances. To capture the typing, we might use GADTs:
\begin{code}
data Eq :: * -> * where
  EqBool :: Eq Bool
  EqList :: Eq t -> Eq [t]
  ...
\end{code}
Class methods would then be defined as GADT pattern matches:
\begin{code}
eq :: Eq t -> t -> t -> Bool
eq d x y = case d of
  EqBool -> not (x `xor` y)
  EqList d' -> ... eq d' ...
  ...
\end{code}
where in the !EqBool! branch we then have access to the type refinement that !t! is !Bool!.

Our actual translation uses the same approach with open data types and open functions, inspired by the Henry Ford encoding of GADTs \citep{McBride00}. We introduce an !EqBool! constructor of the !Eq! type; !EqBool! can construct an !Eq t! of any type !t!, so long as !t! is !Bool!. Our translation of the implementation of equality for Booleans is similarly generic: it is defined for any !a! and value !d! of !Eq a!. In the case that !d! is !EqBool!, pattern matching provides the coercion !h : Bool ~ a!. In translating the body, we discover an apparent type mismatch: !Bool -> Bool -> Bool! is not !a -> a -> Bool!. However, we are able to automatically synthesize the needed coercion.

\subsubsection*{Superclasses}

The translation of the !Ord! class follows a similar pattern.  The superclass can be expressed as just another open function, mapping values of the  !Ord! type to values of the !Eq! type. In translating the instance of the !Ord! class, we provide an instance-specific implementation of the superclass function, following the same typing pattern that we applied in providing instances-specific implementations of the class methods.

\subsubsection*{Method calls}

In a dictionary-passing translation of Haskell type classes, class predicates are interpreted as dictionaries of method implementations. Code that uses type classes is then interpreted to take explicit dictionary arguments, and class methods are selectors for those dictionaries. The interpretation of code using type classes in \SysFD is very similar to the dictionary-passing interpretation. For example, a definition of the $\leq$ operator:
\begin{code}
lte :: forall a. Ord a => a -> a -> Bool
lte = /\ a. \ d x y -> (<) [a] (_ :: Ord a) x y || (==) [a] (_ :: Eq a) x y
\end{code}
would be translated as:
\begin{code}
lte :: forall a. Ord a -> a -> a -> Bool
lte = /\ a. \ d x y. lt [a] d x y || eq [a] (ordEq [a] d) x y
\end{code}
In fact, this code seems identical to the dictionary-passing translation. The difference is in the typing of the instances and class methods. Whereas in dictionary passing translation !d! would contain the implementation of !lt! and the superclass dictionary, and !lt! and !ordEq! would be mere projections, in \SysFD{} !lt! and !ordEq! are open functions and the choice of their implementation is guided by the instance value !d!.

\subsubsection*{Overlapping instances}

Overlapping instances are an occasionally used-feature of the Haskell class system in which the compiler selects among class instances based on ``specificity''.
Perhaps surprisingly, overlapping instances pose no particular challenge for our translation.
In our translation, class method implementations are determined by values representing the instance selected, not (directly) from their types.
Thus, whatever method is used to resolve overlapping instances, the \SysFD image of instance selection will be well-defined.

\subsection{Functional Dependencies}
\label{sec:fundeps}

Functional dependencies were intended to augment type classes with type refinement \citep{Jones00}. In practice, however, their effect was limited to informing type inference, not typing itself. In \SysFD, we can capture the intended behavior of functional dependencies directly. \cref{fig:fundeps-instances} shows the translation of a class with functional dependencies\footnote{\protect\url{Hs/Examples/FunDepsInjective.lean}}, and two of its instances, to \SysFD{}.

\begin{figure}
\begin{minipage}[t]{0.29\textwidth}
\begin{codef}
class F t u
  | t -> u,
    u -> t
\end{codef}
\end{minipage}
\begin{minipage}[t]{0.69\textwidth}
\begin{codef}
open F :: * -> * -> *
open fdFwd :: forall t u v. F t u -> F t v -> u ~ v
open fdBwd :: forall t u v. F t u -> F v u -> t ~ v
\end{codef}
\end{minipage}
\\
\begin{minipage}[t]{0.29\textwidth}
\begin{codef}
instance F Int Bool
\end{codef}
\end{minipage}
\begin{minipage}[t]{0.69\textwidth}
\begin{codef}
instance FIB :: forall a b. Int ~ a -> Bool ~ b -> F a b
instance fdFwd = /\ t u v. \ d_1 d_2.
  guard d_1 is FIB [t] [u] then \ h_1 h_2.
  guard d_2 is FIB [t] [v] then \ k_1 k_2.
    sym h_2 ;; k_2
instance fdBwd = ...
\end{codef}
\end{minipage}
\\
\begin{minipage}[t]{0.29\textwidth}
\begin{codef}
instance F a b =>
  F (Maybe a) (Maybe b)
\end{codef}
\end{minipage}
\begin{minipage}[t]{0.69\textwidth}
\begin{codef}
instance FMM :: forall a b a' b'.
  Maybe a' ~ a -> Maybe b' ~ b -> F a' b' -> F a b
instance fdFwd = /\ t u v. \ d_1 d_2.
  guard d_1 is FMM [t] [u] then /\ a' b'. \ h_1 h_2 e_1.
  guard d_2 is FMM [t] [v] then /\ a'' b''. \ k_1 k_2 e_2.
  let j :: a' ~ a'' = (h_1 ;; sym k_1).2 in
  let c :: b' ~ b'' = fdFwd[a''][b'][b''] (e_1 |> $\Refl{\mathtt{F}}$ @ j @ $\Refl{\mathtt{b'}}$) e_2
  in sym h_2 ;; $\Refl{\mathtt{Maybe}}$ @ c ;; k_2
instance fdBwd = ...
\end{codef}
\end{minipage}
\caption{\SysFD translation of classes and instances with functional dependencies}
\label{fig:fundeps-instances}
\end{figure}

\subsubsection*{Classes.}

First, we consider the translation of the class itself. In addition to introducing an open data type !F! for the class, we introduce two open functions !fdFwd! and !fdBwd! to implement the functional dependencies. Notice that the type of !fdFwd! captures the definition of a functional dependency almost directly: given two instances of !F!, where the determining parameters (!t!) are the same, !fdFwd! provides a coercion equating the determined parameters (!u!).
Similarly, the open function !fdBwd! provides a coercion equating the determined parameter (!t!) with the determining parameter (!u!).

\subsubsection*{Instances.}

Next, we consider a simple instance of !F!, relating !Int! and !Bool!. The translation of the instance itself is immediate, following the pattern we have already established. More interesting is the implementation of the !fdFwd! method. We are given two instances, !d_1! of type !F t u! and !d_2! of type !F t v!. At this point we only have one constructor of !F!, so we guard that both instances were created with that constructor. Doing so brings into scope the type refinements captured by those
constructors, particularly !h_2! of type !Bool ~ u! and !k_2! of type !Bool ~ v!. Composing these coercions gives the required proof that !u ~ v!. We elide the translation of !fdBwd! as it is analogous to the construction of !fdFwd!.

At the bottom of the figure, we consider a less simple instance of !F!, which relates !Maybe a! and !Maybe b! in case !F! already relates !a! and !b!. The translation of this instance includes both the requisite coercions and the assumed instance of !F!. Again, the implementation of !fdFwd! contains the interest. We suppose that both instances !d_1! and !d_2! were constructed by !FMM!; as !FMM! has ``existentially quantified'' type variables !a'! and !b'!, the guards bind those types as well as the coercions. Intuitively, the argument here is that if we have instances !F (Maybe a) (Maybe b)! and !F (Maybe a) (Maybe c)!, then we must also have instances !F a b! and !F a c!, and so that !b ~ c! and !Maybe b ~ Maybe c!. The term spells this argument out in detail; additional complexity arises because of the fact that the first argument to !F! is !Maybe a! is captured by the type refinements !h_1! and !k_1!.

\subsubsection*{Conditions on instances.}

The need to define the !fdFwd! method may seem concerning: how do the instance declarations motivate these definitions? In fact, these are just the semantic realization of what \citet{JonesD08} call the \emph{covering condition} on instances of classes with functional dependencies. They state the covering condition syntactically in terms of the type variables appearing in the class instance. We claim that our version captures the semantic intent of the covering condition: that any two type class predicates satisfied by the same instance of the class must respect the functional dependency. \citet{SulzmannDPJS07} propose a stronger condition, which they call the \emph{coverage condition}, which would reject the instance of !F! in the example. We hope that the \SysFD account of functional dependencies shows that, while the coverage condition may be helpful for type inference with functional dependencies, the relaxed version of the condition proposed by Jones and Diatchki cannot lead to unsoundness, or to violations of the functional dependency.

We have defined two instances of the !fdFwd! method, for the cases where the instances of !F! are built from the same constructor. What about the remaining cases, in which the instances are built from different constructors. Here is (an attempt to define) one such case; the other case is parallel.
\begin{codef}
instance fdFwd = /\ t u v. \ d_1 d_2.
  guard d_1 is FIB [t] [u] then \ h_1 h_2.
  guard d_2 is FMM [t] [v] then /\ a' b'. \ k_1 k_2 e_1.
    -- absurd: h_1 ;; sym k_1 :: Int ~ Maybe a'
\end{codef}
This instance of !fdFwd! is guarded by !d_1! having been constructed via the !F Int Bool! instance, while !d_2! was constructed via the !F (Maybe a) (Maybe b)! instance. Intuitively, this case is impossible to reach: !d_1! and !d_2! seem to place conflicting requirements on the type parameter !t!. And indeed, the typing carries out this intuition: in the body of the definition, we have access to the coercion !h_1 ;; sym k_1! proving that !Int ~ Maybe a'!. Checking for the absurdity of these cases corresponds to the other condition in validating functional dependencies, which both Sulzmann et al. and Jones and Diatchki call \emph{consistency}.

\begin{figure}
\begin{minipage}[t]{0.32\textwidth}
\begin{codef}
instance F (Maybe Int) Int
\end{codef}
\end{minipage}
\begin{minipage}[t]{0.67\textwidth}
\begin{codef}
instance FMI :: forall a b. Maybe Int ~ a -> Int ~ b -> F a b
instance fdFwd = /\ t u v. \ d_1 d_2.
  guard d_1 is FMI [t] [u] then \ h_1 h_2.
  guard d_2 is FMM [t] [v] then /\ a' b'. \ k_1 k_2 e_1.
    -- impossible: need to show u ~ Int /~ Maybe a' ~ v
    -- in a consistent context
\end{codef}
\end{minipage}
\caption{Functional dependencies: erroneous instances}
\label{fig:fundeps-impossible}
\end{figure}

\subsubsection*{Invalid instances.}

A final concern: what if the surface language implementation of functional dependencies were incorrect, and admitted an instance that violated the functional dependency on !F!. Such an instance is shown in \cref{fig:fundeps-impossible}: the instance !F (Maybe Int) Int! conflicts with !F! !(Maybe a)! !(Maybe b)! (when !a! is instantiated to !Int!). While translating the instances itself to \SysFD poses no problem, it is no longer possible to generate well-typed instances for all the possibilities of !fdFwd! and !fdBwd!. The case shown in the figure is guarded by !d_1! having been constructed via the new instance, while !d_2! is constructed via the instance !F (Maybe a) (Maybe b)!. In this case, the context is consistent, but the function needs to return an impossible typing assertion that !Int ~ Maybe a'!. There is no such coercion value: the only value of coercion type is reflexivity (\cref{thm:canonicity}), which will not do. Our current presentation does not rule out divergence at coercion type, and so it is technically possible to provide a diverging term of the required type. However: (i) such a coercion can never lead to a violation of type soundness, exactly because it diverges, and (ii) we imagine that \SysFD programs will be produced by translation from user code, and such a translation would have no reason to generate diverging coercions.

\subsubsection*{Typing with functional dependencies.}

Defining classes with functional dependencies is one thing; using them is another. On the left of \cref{fig:fundeps-polymorphic}, we show a (silly) definition of a function\footnote{\protect\url{Hs/Examples/FunDeps1.lean}} in terms of !F! and its instances. While this definition claims to be polymorphic, considering the instances of !F! suggests that the only possible instantiation of !t! is to !Bool!. In fact, following Jones's original account of principal satisfiable types \citep{Jones95}, on which his account of functional dependencies is based, we would conclude that the types !Bool -> Bool! and !F Int t => t -> t! have the same satisfiable instances, and so ought to be considered equivalent. Nevertheless, neither Hugs nor GHC has ever accepted definitions like this one. In this case of Hugs, this could be considered simply a bug. For GHC, the problem is more foundational: there is no well-typed \SysFC translation of !f!. In \SysFD, however, we do have such a translation, shown on the right of \cref{fig:fundeps-polymorphic}. By applying the !fdFwd! open function to both the known !F Int Bool! instance and the provided !F Int t! instance, we get a coercion !h : Bool ~ t!. This coercion is exactly what we need to cast !not! from !Bool -> Bool! to !t -> t!.

\begin{figure}[ht]
\begin{minipage}[t]{0.33\textwidth}
\begin{codef}
f :: forall t. F Int t => t -> t
f = /\ t. \ d. not
\end{codef}
\end{minipage}
\begin{minipage}[t]{0.65\textwidth}
\begin{codef}
f :: forall t. F Int t -> t -> t
f = /\ t. \ d.
  let h = fdFwd[Int][Bool][t] (FIB[Int][Bool] $\Refl{\mathtt{Int}}$ $\Refl{\mathtt{Bool}}$) d
  in not |> $\Refl\to$ @ h @ h
\end{codef}
\end{minipage}
\caption{\SysFD translation of polymorphic functions with functional dependencies}
\label{fig:fundeps-polymorphic}
\end{figure}

\subsubsection*{Functional dependencies and type families.}

In modern Haskell, type families are preferred to functional dependencies, as they are better supported and more expressive.
We have focused on functional dependencies because our translation fixes an expressiveness gap in other implementations of functional dependencies.
Moreover, our approach to translating functional dependencies applies directly to open and injective type families, with (nearly) identical implementations.
Doing so would realize Constrained Type Families \citep{MorrisE17}, a proposal which repairs some of the surprising features of type families while preserving their expressiveness.
We consider closed type families, and corresponding features in type classes \citep{MorrisJ10,MorrisE17}, an interesting future application of \SysFD.
Even for these cases, we expect their translation to follow the pattern laid out here; the interest would lie primarily in the metatheoretic guarantees.

\subsection{Expressions}

Finally, we describe the translation of expressions.
The surface syntax is almost a subset of \SysFD, and hence the majority of the translation follows directly by structural recursion.
Annotation syntax propagates a type but otherwise does not participate in the translation, nor persist into the translated code.
To handle holes the translation searches over the open functions and their possible instance arguments, identically to instance resolution in Haskell.
This search is non-terminating without some measure on the evidence.
Our examples meet the Paterson Conditions \citep{SulzmannDPJS07}, and so instance selection is guaranteed to terminate.
A hidden requirement of the translation is synthesizing coercions, both for the definitions of instances and to capture the use of functional dependencies in typing terms.
To synthesize coercions, we construct a graph from the equalities in scope, and conduct depth-first search over this graph.
Coercion synthesis may rely on term synthesis to generate pairwise improvements, if there are predicates with functional dependencies in scope.

We would hope that all well-typed terms of the surface language translate to well-typed terms of \SysFD.
As our translation is implemented in Lean, we might even hope to mechanize this result.
However, there are several challenges.
Most significantly, instance selection itself only terminates given non-trivial conditions on surface language programs.
In terms of our mechanization, this means that our translation must be defined as a partial function, and so cannot participate in any non-trivial theorem statements.
Because \SysFD type inference is decidable, however, we are able to verify that the individual outputs of our translator are all well-typed.

\subsection{Open Function Saturation}
\label{sec:instance_saturation}

Just knowing that our translation produces well-typed \SysFD terms may not be enough.
Recall that $\Zero$ inhabits all \SysFD types, so any term could be translated to $\Zero$.
Moreover, $\Zero$ can arise from an invocation of an open function that does not match any of its cases.
Of course, this observation is not a threat to type safety.
Just as divergence, $\Zero$ should be seen as a bottom element of every type; unlike divergence, however, $\Zero$ is detectable and informative.
To assure that our translation is meaningful, then, we should guarantee that the translation of terms is not just well-typed, but will never reduce to $\Zero$.
To do so, we must guarantee that we do not generate a $\Zero$ directly, and that all open functions are defined for all possible well-typed arguments.
We call the latter condition ``saturation'', and will argue that our translation always produces saturated definitions after we have introduced the necessary metatheory \cref{sec:specialization}.

\section{Metatheory}
\label{sec:metatheory}

Standard metatheoretic properties of interest, like progress and preservation, are shown for \SysFD.
Moreover, all of the metatheory is mechanized in a Lean4 development.
Throughout the discussion of metatheory there are mechanization notes that discuss differences between the presentation meant for human consumption and the mechanization.
References to relevant sections of the mechanizations are provided for the stated theorems.

\begin{figure}
  \begin{smalle}
  \[
  \begin{array}{l@{\hspace{5px}}l@{\qquad}l@{\hspace{5px}}l@{\qquad}l@{\hspace{5px}}l@{\qquad}l@{\hspace{5px}}l@{\qquad}l@{\hspace{5px}}l}
    \text{Term variables} & x, y & \text{Kinds} & \kappa  & \text{Terms} & N, M, P, \eta \\
    \text{Type variables} & t, u & \text{Types} & \sigma, \tau
  \end{array}
  \]
  \begin{doublesyntax}
    \text{Contexts} & \mcr{A} & ::= & \square \mid A \Ap N \mid A \TAp \sigma \mid \Cast M A \mid \Case A P M N \mid \Guard A P N  \\
    & & & \mid & \Sym A \mid \Trans A \eta \mid \Trans \eta A \mid \CAp A {\eta} \mid \CAp \eta A \mid \Fst A \mid \Snd A \mid \Univ t \kappa A \mid \CInst A \tau \\
    & & & \mid & \SimC A \eta \mid \SimC \eta A \\
    & \mcr{E} & ::= & \cdots \mid \Choice E M \mid \Choice M E \\
    \text{Term values} & \mcr{V} & ::= & K \mid V \Ap N \mid \LamS x \sigma M \mid V \TAp \sigma \mid \TLam t \kappa M \mid \Refl \tau \mid V_1 \oplus V_2 \\
    \text{Type values} & \mcr{\nu} & ::= & T \mid \nu_1 \Ap \nu_2 \mid \ForallS t \kappa \sigma \mid \Eq {\nu_1} {\nu_2} \kappa
  \end{doublesyntax}
  \end{smalle}
  \caption{Evaluation syntax}
  \label{fig:eval-syntax}
\end{figure}

\begin{figure}
  \begin{smalle}
  \[
    \fbox{$\Gamma \vdash M \leadsto M$}
  \]
  \begin{minipage}[t]{0.45\linewidth}
  \begin{align}
  \Gamma \vdash (\LamS x \sigma M) \Ap N &\leadsto {\Subst x N M}
  \label[rule]{rule:beta}
  \tag{$\beta_\to$}
  \\
  \Gamma \vdash (\TLam t \kappa M) \TAp \sigma &\leadsto {\Subst x \sigma M}
  \label[rule]{rule:beta-type}
  \tag{$\beta_\forall$}
  \\
  \Gamma \vdash \Sym \Refl \tau &\leadsto \Refl \tau
  \label[rule]{rule:refl}
  \tag{$\delta_{\mathsf{refl}}$}
  \\
  \Gamma \vdash \Trans {\Refl \tau} {\Refl \tau} &\leadsto \Refl \tau
  \label[rule]{rule:trans}
  \tag{$\delta_\fatsemi$}
  \\
  \Gamma \vdash \CAp {\Refl \tau} {\Refl \nu} &\leadsto \Refl {\tau \Ap \nu}
  \label[rule]{rule:app}
  \tag{$\delta_@$}
  \\
  \Gamma \vdash \CInst {\Refl {\ForallS t \kappa \tau}} {\nu} &\leadsto \Refl {\Subst x \nu \tau}
  \label[rule]{rule:inst}
  \tag{$\delta_{@[]}$}
  \\
  \Gamma \vdash \Fst {\Refl {\tau \Ap \nu}} &\leadsto \Refl \tau
  \label[rule]{rule:fst}
  \tag{$\delta_{\mathsf{fst}}$}
  \end{align}
  \end{minipage}
  \hspace{0.05\linewidth}
  \begin{minipage}[t]{0.45\linewidth}
  \begin{align}
  \Gamma \vdash \Snd {\Refl {\tau \Ap \nu}} &\leadsto \Refl \nu
  \label[rule]{rule:snd}
  \tag{$\delta_{\mathsf{snd}}$}
  \\
  \Gamma \vdash \SimC {\Refl \tau} {\Refl \tau} &\leadsto \Refl \tau
  \label[rule]{rule:sim}
  \tag{$\delta_\sim$}
  \\
  \Gamma \vdash \Univ t \kappa {\Refl \tau} &\leadsto \Refl {\ForallS t \kappa \tau}
  \label[rule]{rule:forall}
  \tag{$\delta_\forall$}
  \\
  \Gamma \vdash \Cast M \Refl \tau &\leadsto M
  \label[rule]{rule:cast}
  \tag{$\delta_\triangleright$}
  \\
  \Gamma \vdash \Choice \Zero M &\leadsto M
  \label[rule]{rule:zero-l}
  \tag{$\beta_{\Zero\text-1}$}
  \\
  \Gamma \vdash \Choice M \Zero &\leadsto M
  \label[rule]{rule:zero-r}
  \tag{$\beta_{\Zero\text-2}$}
  \\
  \Gamma \vdash A[\Zero] &\leadsto \Zero
  \label[rule]{rule:absorb}
  \tag{$\zeta$}
  \end{align}
  \end{minipage}
  \begin{align}
  \Gamma \vdash \Case {K \Ap {\overline{A}} \Ap {\overline{B}}} {K \Ap {\overline{A}}} M N &\leadsto M \Ap \overline{B}
  \label[rule]{rule:if-hit}
  \tag{$\delta_{\mathsf{if}\text{-}1}$}
  \\
  \Gamma \vdash \Case {K_1 \Ap {\overline{A_1}} \Ap {\overline{B}}} {K_2 \Ap {\overline{A_2}}} M N &\leadsto N
    &&\text{if $K_1 \neq K_2$ or $A_1\neq A_2$}
  \label[rule]{rule:if-miss}
  \tag{$\delta_{\mathsf{if}\text{-}2}$}
  \\
  \Gamma \vdash \Guard {K \Ap {\overline{A}} \Ap {\overline{B}}} {K \Ap {\overline{A}}} M &\leadsto M \Ap \overline{B}
  \label[rule]{rule:guard-hit}
  \tag{$\delta_{\mathsf{guard}\text-1}$}
  \\
  \Gamma \vdash \Guard {K_1 \Ap {\overline{A_1}} \Ap {\overline{B}}} {K_2 \Ap {\overline{A_2}}} M &\leadsto \Zero
    &&\text{if $K_1 \neq K_2$ or $A_1 \neq A_2$}
  \label[rule]{rule:guard-miss}
  \tag{$\delta_{\mathsf{guard}\text-2}$}
  \\
  \Gamma \vdash x &\leadsto \Choice \Zero {\Choice {M_1} {\Choice \ldots {M_k}}}
    &&\text{for all $\InstanceDef x {M_i} \in \Gamma$}
  \label[rule]{rule:open}
  \tag{$\beta_{\mathsf{open}}$}
  \\
  \Gamma \vdash x&\leadsto M
    &&\text{if $\LetDef x M \in \Gamma$}
  \label[rule]{rule:let}
  \tag{$\beta_{\mathsf{let}}$}
  \\
  \Gamma \vdash A[\Choice M N] &\leadsto \Choice {A[M]} {A[N]}
  \label[rule]{rule:map}
  \tag{$\kappa$}
  \\
  \Gamma \vdash E[M] &\leadsto E[N]
    &&\text{if $\Gamma \vdash M \leadsto N$}
  \label[rule]{rule:cong}
  \tag{$\xi$}
  \end{align}
  \end{smalle}
  \caption{Small-step reduction relation}
  \label{fig:reduction}
\end{figure}

\subsection{Reduction and Values}

\subsubsection*{Evaluation contexts and values.}

First, notions of (lazy) reduction and values (in weak head normal form) are required.
\cref{fig:eval-syntax} describes evaluation contexts and values.
A \emph{value} is a constant, the head of an anonymous function ($\LamS x \sigma M$ and $\TLam t \kappa M$), a reflexivity proof ($\Refl \tau$), an application spine where the head is a constant ($K \Ap \overline{A}$), or a choice operator between values ($V_1 \oplus V_2$).
As expected for a definition of values it will be the case that they do not reduce.

\subsubsection*{Evaluation}

Figure~\ref{fig:reduction} describes the full small-step reduction relation.
Reduction is defined in terms of an environment $\Gamma$, which will provide definitions of open functions and \texttt{let}-bound variables.
A single term $M$ can reduce to multiple possible choices.
For example, the reduction rule for instance instantiation of an open variable reduces to a term representing all possible instance choices.

The anonymous function $(\beta)$ rules are entirely standard.
As the only value of coercion type is $\Refl\tau$ (or rather a potential tree of reflexivity proofs), the reduction rules for coercions $(\delta)$ manipulate only the coercion's type annotation.
The reduction rule for casts $(\delta_\triangleright$) eliminates fully-reduced coercions.
The reduction rules for conditionals and guards $(\delta_{\mathsf{guard}})$ check whether the scrutinee has the same term constant and a shared prefix application spine.
A guard that does not match ``vanishes'' from evaluation; we represent this as the monoid identity element of choice.
Instance invocations $(\beta_{\mathsf{open}})$ evaluate to all the definitions of the instance in scope folded by the monoid formed using $\Zero$ and $\oplus$; in contrast, traditional !let!-bound variables have a single definition.

Evaluation contexts are defined in Figure~\ref{fig:eval-syntax} for lazy reduction.
There are two kinds of evaluation context an ``absorptive'' context (represented by $A$) and a full evaluation context (represented by $E$).
The notation $E[M]$ means substituting the hole in the evaluation context (represented by $\square$) with the associated term $M$.
The absorptive context is used to capture how reduction interacts when the choice operators are not part of the ambient context.
The absorption rule $(\zeta)$ defines how the identity element $\Zero$ absorbs surrounding syntax.
The mapping rule $(\kappa)$ defines how absorption contexts map over the choice operator.
We have a standard congruence rule $(\xi)$ for evaluation contexts.

\subsubsection*{Mechanization note.}

We have presented our syntax (\cref{fig:syntax,fig:eval-syntax}) and rules (\cref{fig:env-kind-typequ,fig:typing,fig:typing-coercions,fig:typing-declarations}) split across the syntactic categories of kinds, types, and terms.
To work around limitations in Lean4's support for mutually inductive data types, our mechanization merges syntactic categories in one inductive type, and merges the typing judgments in another.
This also means that our mechanization has a single notion of value, across types and terms.

\subsection{Properties of Typing}
As promised, the first important theorem is that values are sound.
That is, if a well-typed term is a value then it does not reduce.
This means the characterization of values is correct with respect to reduction.
Moreover, it also means that a value is not $\Zero$, but it does allow for $\Zero$ to appear as a subexpression inside of a value.

\begin{theorem}[Soundness of Values\protect\footnote{\protect\url{SystemFD/Metatheory/Progress.lean:val_sound}}]
  If $\TypeJ \Gamma V \tau$ then $\not \exists t^\prime.\ \Gamma \vdash V \leadsto t^\prime$ and $V \neq \Zero$
\end{theorem}

Additionally, it is simple to show that types are always values.
This is perhaps unsurprising as there is no reduction in types.

\begin{theorem}[Types are Values\protect\footnote{\protect\url{SystemFD/Metatheory/Progress.lean:types_are_values}}]
  If $\Gamma \vdash \tau : \kappa$, then $\tau$ is a value
\end{theorem}

Most similar systems without subtyping have the property that types are unique.
With the existence of $\Zero$ as a formal syntactic bottom element this property is trivially refuted in general, but there are a few instances of uniqueness of types that are still obtainable.
The most readily available is uniqueness of types in the absence of $\Zero$ as a subexpression.

\begin{theorem}[Uniqueness of Types Modulo $\Zero$\protect\footnote{\protect\url{SystemFD/Metatheory/Uniqueness.lean:uniqueness_modulo_zero}}]
  If $\Gamma \vdash M : \tau$, $\Gamma \vdash M : \sigma$, and $0$ is not a subexpression of $M$ then $\tau = \sigma$
\end{theorem}

Additionally, if a term has a neutral form (i.e., is an application spine with a variable or a constant at the head) then it also has a unique type.
Using the prior theorem makes this clear: $\Zero$ is not a subexpression of a variable.
Meaning, the function type of the head is unique and fully determines the return type, regardless of the status of the arguments.
Note we state the theorem below relative to term applications, but the theorem holds more generally relative to arbitrary applications (both intermixed term and type).
These two forms of uniqueness are sufficient to prove the remaining metatheoretic theorems of interest.

\begin{theorem}[Uniqueness of Types With Neutral Form\protect\footnote{\protect\url{SystemFD/Metatheory/Uniqueness.lean:uniqueness_modulo_neutral_form}}]
  If $\Gamma \vdash M : \tau$, $\Gamma \vdash M : \sigma$, and $M = x\ M_1\ \cdots\ M_n$ then $\tau = \sigma$
\end{theorem}

Previously we claimed the only values at coercion type are trees of reflexivity proofs.
To make this clear, a value at coercion type can be one of two things: either $\Refl \tau$ or a choice between two values, $V_1 \oplus V_2$.
However, the choice operator maintains the same type in both the left and right possibilities, meaning both $V_i$ are values are coercion type.
By induction the $V_i$'s must both be a tree of reflexivity proofs.
Note that this works for any well-typed term $\eta$ in any context $\Gamma$ because being a value precludes a neutral form.
Meaning, it is not possible for $\eta$ to simultaneously be a value and to be a lambda bound variable.

\begin{theorem}[Canonicity at Coercion Type\protect\footnote{\protect\url{SystemFD/Metatheory/Canonicity.lean:refl_is_val}}]
  \label{thm:canonicity}
  If $\TypeJ \Gamma V : \Eq \sigma \tau \kappa$, then $V = \Refl \tau$ or there exists terms $M_1$ and  $M_2$ such that $V = M_1 \oplus M_2$
\end{theorem}

An immediate corollary to this theorem is that both $M_i$'s must be values by inspection of the evidence that $V$ is a value.
Additionally, we can prove a similar canonicity theorem for function types.
Unlike with coercions, a function type has three possibilities: a lambda binder, a choice of function type values, and a constructor constant.
Just like with coercion types, the potential tree of possibilities consists of leafs that are values.

\begin{theorem}[Canonicity at Function Type\protect\footnote{\protect\url{SystemFD/Metatheory/Canonicity.lean:canonical_lambda}}]
  If $\TypeJ \Gamma V : \sigma \to \tau $, then $V = \LamS x \tau M$, or $V = K$ for a constructor term constant $K$, or there exists terms $M_1$ and  $M_2$ such that $V = M_1 \oplus M_2$
\end{theorem}

\subsubsection*{Mechanization note.}

Because our mechanization merges syntactic classes, the theorems above are not as trivial as they might seem.
Critically, it is also the case that well-typed syntax may be reinterpreted into its distinct categories via a property called \emph{classification}.
Let $t, A$ range over the union of terms, types, kinds, and an additional symbol $\square$ introduced to classify kinds.
Let $-\vdash-:-$ be the union of the typing and kinding relations, along with a new axiom $\Gamma \vdash \kappa : \square$ for kinds $\kappa$.

\begin{theorem}[Classification\protect\footnote{\protect\url{SystemFD/Metatheory/Classification.lean}}]
  If $\Gamma \vdash t : A$ then $A$ satisfies one and exactly one of the following:
  \begin{enumerate}
    \item $A = \square$
    \item $\Gamma \vdash A : \square$
    \item $\Gamma \vdash A : \kappa$ for $\Gamma \vdash \kappa : \square$
  \end{enumerate}
\end{theorem}

\noindent
Without this property the intrinsic presentation would not coincide with the extrinsic one on well-typed syntax.
Thus, it is critical this property holds in order to allow a presentation such as the one in Section~\ref{sec:system-fd}.

\subsection{Properties of Reduction}

A significant motivation for \SysFD is to guarantee it is well-behaved and moreover that there is high assurance in the proofs that it is well-behaved.
In this context, well-behaved means type soundness, or that reduction of terms does not get stuck.
The operational way of accomplishing this goal is proving both progress and preservation.

For \SysFD, \emph{progress} states that a well-typed term is either a value, 0, or that there is some term that it reduces to assuming the context does not contain bound variables introduced by an abstraction.
Typically, this requirement is communicated by requiring the context to be empty, but in \SysFD declarations, including open functions, are part of the context.
For that reason, the context needs to be retained and instead specifically $\lambda$-bound variables must be excluded.

\begin{definition}[$\lambda$-free context]
  Context $\Gamma$ is $\lambda$-free iff $x : \sigma \in \Gamma$ implies that $\LetDef x M \in \Gamma$.
\end{definition}

\noindent
Of course, there are other approaches to this problem, such as a mutual inductive definition of values and neutrals.
This approach is taken because of its simplicity and because it correlates with expectations of type soundness for non-strict languages.

\begin{theorem}[Progress\protect\footnote{\protect\url{SystemFD/Metatheory/Progress.lean:progress}}]
  Suppose $\Gamma$ is $\lambda$-free.
  If $\TypeJ \Gamma M \tau$ then either $M$ is a value, or $M = \Zero$, or there exists $M^\prime$ such that $\Gamma \vdash M \leadsto M^\prime$
\end{theorem}

Recall that \SysFD has two bottom elements: the standard bottom element of non-termination which is not syntactically expressible, and a detectable bottom element $\Zero$.
In reality, $\Zero$ could be considered a value without the metatheory changing in any significant way, but conceptually we believe it to be better to understand $\Zero$ as a bottom element instead of a value since it does inhabit all types.
However, note that the existence of $\Zero$ in no way challenges or effects the syntactic type safety of \SysFD.
An arbitrary coerce function can be constructed with $\Zero$ just as readily as it can be constructed with a recursive open function, there is no claim here that the type theory of \SysFD is a consistent logic.

Finally, the remaining theorem for proving type safety of \SysFD is preservation.
This theorem is stated and proved in an entirely standard way.
The combination of progress and preservation gives \SysFD the property of syntactic type safety.
Meaning, any well-typed term reduces to another well-typed term and at any point in the chain we know exactly what that term means, there are no surprises.
This is, of course, relative to the reflexive-transitive closure of our small-step reduction relation which models a lazy language.

\begin{theorem}[Preservation\protect\footnote{\protect\url{SystemFD/Metatheory/Preservation.lean:preservation}}]
  If $\TypeJ \Gamma M \tau$ and $\Gamma \vdash M \leadsto M^\prime$ then $\TypeJ \Gamma {M^\prime} \tau$
\end{theorem}

\subsection{Specialization of the Translation}
\label{sec:specialization}

The translation from our surface language to \SysFD characterizes a subset of \SysFD terms.
In particular, the subset satisfies the following definition.

\begin{definition}
  A \SysFD term is Haskell-Style with Statically Determined Instances (HSSDI) if, and only if:
  \begin{enumerate}
    \item Every open function is well-founded by a measure on the type instance evidence;
    \item Every invocation of an open function is supplied with a concrete instance value for the open type arguments or it is a variable bound to a guard pattern; and,
    \item For every permutation of applied instance arguments to an open function there exists an instance that has a preamble of guards that match exactly on that permutation.
  \end{enumerate}
\end{definition}

\begin{theorem}
  Let $\llbracket S \rrbracket = M$ be the partial translation of a surface language term $S$ to a \SysFD term $M$.
  If $M$ is defined, then $M$ is HSSDI.
\end{theorem}
\begin{proof}
  Conditions two and three are testable by syntactic inspection.
  Traverse the term $M$ for every open function applied to arguments.
  Check that the open functions have a static open type instance applied for every open type argument or that
  it is a variable bound by a guard pattern. Hence, the second condition is satisfied.

  Now, for every instance of every open function scan it for a preamble of guards.
  A preamble consists of guards potentially separated by lambda binders.
  Extract only the patterns of these guards into a tuple. If the permutation of type instances
  matches any tuple of guard patterns then the third condition is satisfied.

  For the first condition the translation uses a naive metric consisting of the size of instance evidence.
  However, this is sufficient for our examples and does provide a well-founded restriction on open function recursion.
\end{proof}

There is a syntactic criterion for excluding reduction to zero that we will utilize.
Imagine a \SysFD term that does not contain an open function variable, a guard, or a zero.
Then it is impossible for it to produce a zero, because only open function reductions and guard reductions can possibly introduce a zero into the term.
This is almost true, but we must also eliminate let-bound variables to prevent substitution of terms that do not satisfy this property.

\begin{theorem}[Syntactic Guarantee of No Zeroes\protect\footnote{\protect\url{SystemFD/Metatheory/HSSDI.lean:syntactic_guarantee_of_reduction_to_zero_impossible}}]
  \label{thm:synnozero}
  Let $M$ be a \SysFD term such that there are no guards, zeroes, open function variables, or let-bound variables appearing as subexpressions in $M$, then $M$ does not reduce to $0$
\end{theorem}

Finally, we approach the problem by considering specializing a term such that it satisfies the above syntactic condition.
This specialization highlights another important aspect of our semantics of Haskell in terms of \SysFD.
The intent is exactly for the additional \SysFD machinery to be an intermediary.
Hence, the observation that we can specialize away guards and open functions leaves us with an interpretation of Haskell's typeclasses that does not incorporate any unusual bottom elements or operations.

\begin{theorem}[Specialization]
  If $M$ is HSSDI then there exists a specialized term $M^\prime$ that does not reduce to 0
\end{theorem}
\begin{proof}
  As a preprocessing step, we substitute all let-bound variables into the term; because let definitions are not recursive this process is well-founded.
  We proceed by an analogy to optimization and in particular specialization.
  Consider a HSSDI term, if we unfold an open function and apply its arguments then there is a tree of choices for each instance.
  However, only the instances that match the patterns will survive, of which there must be at least one by the third condition of HSSDI.
  As the bodies of the guards are also HSSDI, the process may be repeated by induction on the size of the instance evidence, any open functions we unfold inside the body must have a smaller instance evidence.
  Using the absorption rules of reduction any remaining zeroes are eliminated.
  Now the term does not consist of any open function variables, guards, zeros, or let-bound variables.
  Let $M^\prime$ be the term produced by these transformations.
  Then, by Theorem~\ref{thm:synnozero} $M^\prime$ does not reduce to 0.
\end{proof}

Specialization of a term, however, does not rule out the case where more than one valid instance is chosen.
This is effectively a failure of coherence which in our semantics manifests as exploring every possible instance that may work.
Alternatively, the HSSDI condition could demand that there is exactly one instance that matches, enforcing a standard coherence condition.

\subsubsection*{Mechanization note.}
The reader at this point may wonder why the above theorems are not mechanized if they are seemingly intuitive.
The issue is three-fold.
First, implementing the translation in a way that is amenable to proving properties is a hard problem, and convincing Lean4 that the various components are terminating is challenging.
Second, the first condition of HSSDI is intuitively clear because we know concretely how evidence is going to be placed, but independently characterizing this condition on arbitrary \SysFD terms is difficult.
Third, related to the previous problem, the specialization proof is inherently a termination proof with respect to a new reduction relation captured the transformations of specialization.
Termination arguments are usually hard to mechanize and in our experience this example was no exception.

\subsection{Emulating Autosubst in Lean4}
The Autosubst library in Rocq \cite{SchaferTS15} and Autosubst2 external tool \cite{StarkSK19} are, in the authors opinion, well designed methods of dispensing with the bureaucracy of substitution.
However, both the library and the external tool target only Rocq.
There has been work on reimplementing the library of Autosubst in both Lean3 and Lean4, but none that has been published or gained any traction to the authors knowledge.
In this work, explicit tactics (i.e., \texttt{asimpl}) are replaced with Lean4's simplification machinery.
This gives a similar feel to the Autosubst developments in Rocq where wrangling substitutions becomes a simple matter of invoking one tactic: \texttt{simp} (\texttt{asimpl} respectively. in Rocq).
There are two main differences between this work and other works either directly using Autosubst or emulating it:
\begin{enumerate}
  \item this work uses type classes without additional macro syntax in describing the inductive type to convey the necessary structure;
  \item and substitutions are generalized to a map from variables (i.e., natural numbers) to a \emph{substitution action} consisting of either a renaming or a replacement.
\end{enumerate}
The first difference is neither better nor worse, it is merely a different exploration of the library design space.
To accomplish it, a user of the library must provide a \emph{substitution map} which behaves similarly to a standard functorial map with the exception of applying lifting at the associated places where a binder would expand the scope.
Afterwards, the functorial laws must be proven relative to substitutions: that the identity substitution behaves as an identity; and that substitutions compose.
The later proof is typically difficult, but an associated tactic is able to solve the theorem for any non-indexed extrinsically presented syntax.

Generalizing the type of substitutions enables variables to have associated data without sacrificing the benefits of the Autosubst design.
In this way, a variable may be marked as either a term variable or a type variable while still having an extrinsic presentation of syntax.
For example, if it became necessary to write an algorithm to decide the classification of a term then such a function would be possible \emph{without} carrying around a context to associate variables to syntactic categories.
In essence, it would mean that the intrinsic nature of the classification theorem would be present in the extrinsic syntax.
This is a complementary approach to multi-sorted substitutions that Autosubst prefers when syntax is naturally split into separate inductive types based on syntactic category.

Reuse of the library provided and developed in the mechanization of \SysFD requires copying one library file\footnote{\protect\url{SystemFD/Substitution.lean}} and providing the necessary type class instances\footnote{\protect\url{SystemFD/Term/Substitution.lean:substTypeLaws_Term}}.

\section{Related Work}
\label{sec:related}

We summarize three areas of related work: we compare \SysFD to \SysFC, Haskell's core language; we relate our semantics to other work on type families and functional dependencies, and finally we relate our work to other semantics of overloading.

\subsection{\SysFCbf}

\SysFC is a core language for Haskell. Initially, \SysFC combined System F with first-class type equality proofs \citep{SulzmannCJD07}, used to implement (among other things) associated types \citep{ChakravartyKPJM05} and GADTs \citep{SchrijversPJSV09}. Later extensions to \SysFC include kind promotion and polymorphism \citep{YorgeyWCJVM12}, enabling flexible type-level programming with singleton types \citep{EisenbergW12}, and kind equalities \citep{WeirichHE13}, allowing all the features of Haskell's type system to be used at the kind level as well.

The treatment of coercions in \SysFD is immediately inspired by that of \SysFC, but differs in several ways. \SysFC treats coercions themselves as types, not as terms, and so coercion types are kinds, not types. This has several consequences. Type erasure---that is, the property that type information can be erased without changing the semantics of a \SysFC term---is automatically extended to coercion erasure as well. Moreover, because \SysFC has no fixed point operation at the level of types, there are no non-terminating coercions. \SysFD treats coercions as terms, not as types, as our interpretation of functional dependencies and type functions will rely on open functions that return coercions. We consider ruling out divergence at coercion type as interesting but essentially orthogonal future work.

The semantics of coercions in \SysFC is captured by a collection of ``push'' rules, which move coercions out of the way of reductions in terms. In contrast, the semantics of \SysFD simply relies on reducing coercions to reflexivity, and then eliminating them. We believe this approach could not apply in \SysFC because the interpretation of type families in \SysFC relies on introducing non-reflexivity coercion axioms, which do not have elimination forms.

\SysFD and \SysFC have very different ideas of how to express source language features.
\SysFC reflects many source language features directly, or nearly directly.
For example, \SysFC has axioms that exactly parallel the form of closed type families.
In \SysFD, in contrast, our goal is to express many surface language features by translation to a single set of core features.
These approaches impose different metatheoretic demands.
For \SysFC, adding new features can mean extended the core type system, and so reproving type safety.
For \SysFD, encoding new features requires proving saturation.
Our hope is that the additional expressiveness of \SysFD will offset the metatheoretic requirements of its use.

\subsection{Type Families and Functional Dependencies}

Type families were originally tied to type classes as associated types \citep{ChakravartyKPJM05}, but were soon freed from their tie to the class system \cite{SchrijversPJCS08}. Later extension include closed type families \citep{EisenbergVPJW14}, which allow overlapping type family instances so long as they are explicitly ordered, and injective type families \citep{StolarekJE15}, which add more of the flavor of functional dependencies to type families.

Each of these features are given semantics by corresponding extension of the structure of \SysFC coercion axioms. Open type families introduce axioms asserting that the left- and right-hand sides of their defining equations are equal. Closed type families introduce chains of axioms, each of which applies only to types apart from the earlier axioms in the chain. Injective type families introduce injectivity axioms. In each case, then, these axioms must be shown not to introduce unsoundness in \SysFC; the example of injectivity introducing unsoundness~\cite{StolarekJE15} shows the kind of complexity that needs to be addressed in these proofs.

\citet{MorrisE17} introduce \SysCFC, which requires that uses of type families in types must be justified by the axioms present in the program. This (seemingly minor) restriction allows them to lift side conditions on the definition of closed type families, and gives a significantly simpler proof of type safety. They identify the side conditions on injective type families as also likely to be addressed by their technique, but do not formalize injective type families in \SysCFC.

Functional dependencies were originally introduced by \citet{Jones00} as a way for type classes to give rise to type refinement. Much subsequent work on functional dependencies \citep{DuckPSS04,SulzmannDPJS07,JonesD08} have focused on the role of functional dependencies in type inference, not on their semantics or otherwise on their representation in a core language. \citet{KarachaliasS17} give a translation of functional dependencies into type functions; to our knowledge, this is the only prior work that associates functional dependencies with coercions. Of course, it neither simplifies nor complicates the underlying semantics of type families.

The \SysFD interpretation of type families and functional dependencies is fundamentally different from the prior work, as it does not depend on extension of \SysFD's coercion type, and so cannot compromise type safety. Instead, the burden is placed on the translation to be able to generate well-typed open function instances that return suitable coercions.

\subsection{Other Treatments of Haskell-style Overloading}

\citet{Jones95} introduced \emph{improvement} to reflect the relationship between the predicate satisfiability and polymorphism. On his account, two type schemes ought to be considered equivalent if their satisfiable ground instances---that is to say, the ground instances such that any predicates are provable---are the same. Jones's explanation of functional dependencies \citep{Jones00} is in terms of improvement. However, he considers the consequences of improvement for type inference, not for typing itself. Recall our example class !F! with instance !F Int Bool!~\cref{fig:fundeps-polymorphic}; while Jones's theory holds that the types !F Int a => a -> a! and !Bool -> Bool! are equivalent, and thus a type inference algorithm would be justified in producing either type, we do not believe he ever described a typing rule (or its semantics) that would allow a term of one type to be used as a term of the other.

\citet{Morris14}, building on \citet{Ohori89} and \citet{Harrison05}, described a domain-theoretic semantics of Haskell type classes in which class methods are interpreted as sets of their implementations at each ground type. This approach naturally realizes Jones's account of satisfiability. Morris proposes a typing rule that would allow !not! to be used at type !F Int a => a -> a!, and proves its soundness. However, Morris's account is purely semantic, and so does not explain how this typing rule could be realized in a core language like \SysFC.

\section{Future Work}
\label{sec:conclusion}

We have described \SysFD, a new core language for Haskell, and shown that many of the advanced features of Haskell's type class and type family systems can be naturally expressed in \SysFD. We conclude by discussing three areas of future research.

\subsubsection*{Formalizing source type systems}

Existing work on type families in Haskell has formally described only the core language, not the typing of the surface language or the translation from the surface language to the core language. We suspect several reasons for this. On the one hand, Haskell has a rich surface language, with many features irrelevant to the type class and type family systems. On the other, the surface language features of type families are represented fairly directly in their \SysFC encodings.

That said, we see several benefits to formally (and indeed mechanically) characterizing type classes and type families in Haskell's source language and  their translation to \SysFD. This will allow us formally verify that \SysFD accounts for the features of Haskell, and moreover that the source language type system is sufficient to guarantee that the translation to \SysFD is well-defined. More ambitiously, this will provide a basis for giving semantics directly to source programs (perhaps following Morris's domain theoretic approach \citep{Morris14}) and showing that this semantics is preserved by the translation to \SysFD.

\subsubsection*{Compiling System $F_D$}

\SysFC gives very pleasant compilation story for type classes and type families. Type classes are source language ephemera, transformed into dictionaries en route to \SysFC. Type families are realized by coercions, which are included in \SysFC's type erasure results. Thus, an implementer can be sure that type families, as interpreted by \SysFC, need introduce no runtime overhead of their own.

We believe that coercions should be erasable in \SysFD as well. The most significant challenge is that, without any restriction on the form of instances, open functions introduce arbitrary recursion. On the other hand, restricting the form of open function instances will correspondingly restrict the form of type class and type family instances that can be translated to \SysFD. While Haskell formally has quite strict rules for type class instances, which would undoubtedly allow sufficient restriction on open functions to guarantee termination, in practice Haskell developers rely on relaxing such restrictions.

\subsubsection*{Richer type system features}

Functional dependencies \citep{Hallgren01} and type families \citep{KiselyovJS10} can both be used to express computation at the type level. This has naturally led to encodings of patterns from dependently typed programming in Haskell \citep{LindleyM13}, as well as a gradual evolution of Haskell itself toward dependent types \citep{WeirichVAE17,WeirichCVE19}. A natural question is whether type classes and type families are still relevant in a dependently-typed Haskell. We suggest that they will be: type classes and type families are fundamentally about openness, which is not a typical concern in dependent type theory. Moreover, we argue that our treatment of coercions is closer to the typical treatment of the equality type in dependent type theory (i.e. Martin-L\"of's $J$ combinator) than the ``push'' rules of \SysFC. That being said, we think there would be significant challenges to adapting \SysFD to dependent types, most immediately the presence of open function calls in types.

\ifreview
\else
\section*{Data Availability Statement}
\label{sec:data}
\begin{acks}

\end{acks}
\fi

\InlineOff
\bibliographystyle{ACM-Reference-Format}
\bibliography{first}

\ifextended

\clearpage

\appendix
\fi

\end{document}